\documentclass[submission,copyright]{eptcs}
\providecommand{\event}{TARK 2017} 
\usepackage{underscore}           
\pdfoutput=1

\usepackage{booktabs} 

\usepackage[square,sort,comma,numbers]{natbib}
\usepackage[table]{xcolor}
\usepackage{mathtools} 
\usepackage{amsmath,amssymb,amsthm,microtype,booktabs,wrapfig}
\usepackage{tikz}
\usetikzlibrary{shapes, calc, positioning}

\newtheorem{theorem}{Theorem}
\newtheorem{corollary}[theorem]{Corollary}
\newtheorem{proposition}[theorem]{Proposition}
\newtheorem{lemma}[theorem]{Lemma}
\theoremstyle{definition}
\newtheorem{definition}{Definition}

\newcommand{\pref}{\succcurlyeq}
\newcommand{\rev}{^{\text{rev}}}

\renewcommand*{\ge}{\geqslant}

\renewcommand{\a}{\ensuremath{a}}
\renewcommand{\b}{\ensuremath{b}}
\renewcommand{\c}{\ensuremath{c}}
\renewcommand{\d}{\ensuremath{d}}

\newcommand{\yes}[1]{{\color{black!60!green}{\textbf{\textit{#1}}}}}
\newcommand{\no}[1]{{\color{black!70}#1}}

\def\titlerunning{Condorcet's Principle and the Preference Reversal Paradox}
\def\authorrunning{D. Peters}

\begin{document}
\title{Condorcet's Principle and the Preference Reversal Paradox}  
\author{Dominik Peters
  \institute{Department of Computer Science \\ University of Oxford}
	\email{dominik.peters@cs.ox.ac.uk}}

\maketitle

\begin{abstract}
We prove that every Condorcet-consistent voting rule can be manipulated by a voter who completely reverses their preference ranking, assuming that there are at least 4 alternatives. This corrects an error and improves a result of [Sanver, M. R., \& Zwicker, W. S. (2009). One-way monotonicity as a form of strategy-proofness. Int J Game Theory 38(4), 553-574.] For the case of precisely 4 alternatives, we exactly characterise the number of voters for which this impossibility result can be proven. We also show analogues of our result for irresolute voting rules. We then leverage our result to state a strong form of the Gibbard--Satterthwaite Theorem.
\end{abstract}

\section{Introduction}
\label{intro}
The Gibbard--Satterthwaite Theorem establishes that every non-trivial voting rule can be manipulated by voters through misrepresenting their preferences. In this paper, we will see that Condorcet extensions (voting rules that select the Condorcet winner if one exists) suffer from a particularly offensive failure of strategyproofness: all of them can be manipulated by a voter who completely reverses their preference ranking. For example, such a voting rule might designate $c$ to be the winning alternative if voter $i$ truthfully reports the ordering $a \succ_i b \succ_i c \succ_i d$, but choose $b$ as the winner if voter $i$ instead reports the ordering $d \succ_i c \succ_i b \succ_i a$. Since $i$ truthfully prefers $b$ to $c$, this is a successful manipulation, which one might consider surprising given that $i$ misreported every possible pairwise comparison. We will say that voting rules that are manipulable in this way suffer from the \emph{preference reversal paradox}. While all Condorcet extensions exhibit this paradox, scoring rules (such a plurality and Borda's rule) are immune.

Preference reversal paradoxes were first introduced by Sanver and Zwicker~\cite{SaZw09a} in their study of monotonicity properties; they say that voting rules which avoid this paradox satisfy \emph{half-way monotonicity}.\footnote{They chose this name because half-way monotonicity is a weaker version of their notion of \emph{one-way} monotonicity.}
As Sanver and Zwicker~\cite{SaZw09a} show, half-way monotonicity is a weaker property than \emph{participation}, an axiom stating that a voter cannot obtain a strictly better result by abstaining from an election; equivalently, participation says that voting truthfully guarantees a (weakly) better result than not voting at all. In a famous paper, Moulin \cite{Moul88b} showed that participation is incompatible with Condorcet-consistency, so that Condorcet extensions must suffer from the \emph{no-show paradox} [Fishburn and Brams, \citealp{BrFi83a}]. This result is often interpreted as showing that all Condorcet extensions are \emph{manipulable} (through abstention). Notice, however, that this notion of manipulation (referring to electorates of different sizes) is quite different from the fixed-electorate manipulations that are the subject of the Gibbard--Satterthwaite Theorem, where a voter changes their preference ordering in some way \citep[see also N\'u\~nez and Sanver,][]{NuSa17a}. We will see that half-way monotonicity, which is both weaker than participation and weaker than strategyproofness in the Gibbard--Satterthwaite sense, is already incompatible with Condorcet-consistency.

This result first appeared in Sanver and Zwicker~\cite{SaZw09a} who gave a proof that, for 4 or more alternatives and for sufficiently many voters, Condorcet extensions must fail half-way monotonicity. However, their proof contains an arithmetical mistake\footnote{In the last paragraph of the proof or their Theorem 5.2, they calculate that $n^*(Q) = 30 + 8$, when in fact $n^*(Q) = 30 + 4 \cdot m! \gg 38$ which makes their ``Condition M'' inapplicable to profile $Q$. This problem was noticed by Wei Yu and Tokuei Higashino (Zwicker, private communication).} that is non-trivial to fix. The proof technique also is only able to establish an impossibility for electorates containing a sufficiently large \emph{even} number of voters. Further, their proof requires at least 702 voters to go through, this bound growing exponentially as the number of alternatives increases,\footnote{The large number arises because the proof uses several copies of the full profile containing a copy of each of the $m!$ preference orders. Fixing the arithmetical error described above tends to necessitate using many more voters than this (Zwicker, private communication).} which leaves open the question of whether the preference reversal paradox is a problem in practical voting situations with moderate numbers of voters.

We give a direct proof of the impossibility, treating the cases of electorates with odd and even numbers of voters separately. Our arguments require 15 voters for the odd case and 24 voters for the even case. These constant bounds hold for any number $m\ge 4$ of alternatives. Using computer-aided techniques we are able to show that these results are tight: for the case of precisely 4 alternatives, there exist Condorcet extensions satisfying half-way monotonicity for up to 13 voters and 22 voters, respectively. (For 3 alternatives, it is known that the maximin rule with some fixed tie-breaking is a Condorcet extension satisfying half-way monotonicity.)

Both our positive and our negative results were proved with the help of SAT solvers, using a technique introduced by Geist and Endriss~\cite{GeEn11a} and Tang and Lin~\cite{TaLi09a}. For a recent survey, see the book chapter by Geist and Peters~\cite{GePe17a}. The general approach is to produce a formula of propositional logic whose models correspond to voting rules that are Condorcet-consistent and half-way monotonic. We can then pass this formula to a SAT solver. If the formula is satisfiable, we obtain an example of a good voting rule; if it is unsatisfiable, we have an impossibility result. In the unsatisfiable case, using an idea of Brandt and Geist~\cite{BrGe15a}, we can then extract a \emph{minimal unsatisfiable set} (MUS) which can often be translated into a human-readable impossibility proof. This technique was used to prove impossibility results for Fishburn-strategyproofness in majoritarian social choice functions \citep[Brandt and Geist,][]{BrGe15a}, for Fishburn-participation in the same setting \citep[Brandl et al.,][]{BBGH15a}, for the no-show paradox \citep[Brandt, Geist, and Peters,][]{BGP16c}, and for probabilistic social choice rules \citep[Brandl et al.,][]{BBG15a}. Since our techniques are variations of the technique of Brandt, Geist, and Peters~\cite{BGP16c}, we will only give a brief overview of the method in this paper.

\begin{table*}[t]
	\colorlet{works}{green!35!black!15!white}
	\colorlet{paradox}{red!35!black!70!white}
	\setlength{\tabcolsep}{4.35pt}
	\centering
	\begin{tabular}{cccccccccccccccccccccccccc}
		\multicolumn{2}{r}{$n=$\!\!}
		& 3 & 4 & 5 & 6 & 7 & 8 & 9 & 10 & 11 & 12 & 13 & 14 & 15 & 16 & 17 & 18 & 19 & 20 & 21 & 22 & 23 & 24 & 25 & 26\\ \toprule
		
		\multicolumn{11}{l|}{\cellcolor{works}participation
		}
		& \multicolumn{15}{l|}{\cellcolor{paradox}
		} \\
		
		\midrule
		
		\multicolumn{14}{l|}{\cellcolor{works}half-way monotonicity}
		& \multicolumn{1}{l|}{\cellcolor{paradox}}
		& \multicolumn{1}{r|}{\cellcolor{works}}
		& \multicolumn{1}{l|}{\cellcolor{paradox}}
		& \multicolumn{1}{r|}{\cellcolor{works}}
		& \multicolumn{1}{l|}{\cellcolor{paradox}}
		& \multicolumn{1}{r|}{\cellcolor{works}}
		& \multicolumn{1}{l|}{\cellcolor{paradox}}
		& \multicolumn{1}{r|}{\cellcolor{works}}
		& \multicolumn{4}{l|}{\cellcolor{paradox}}
		\\
		\bottomrule \\ \addlinespace[-7pt]
	\end{tabular}
	\setlength{\tabcolsep}{5.1pt}
	\begin{tabular}{lllll}
		\hspace{9.7cm} &
		\cellcolor{works}{\quad} & Possibility\: &
		\cellcolor{paradox}{\quad} & Impossibility
	\end{tabular}
	\caption{Numbers $n$ of voters for which Condorcet extensions can satisfy participation or half-way monotonicity, when there are exactly $m = 4$ alternatives.}
	\label{tbl:results}
\end{table*}

As mentioned above, our theorem implies Moulin's result for participation. Brandt, Geist, and Peters~\cite{BGP16c} recently showed that, for 4 alternatives, Moulin's impossibility requires 12 voters to go through, while there exists a Condorcet extension satisfying participation for up to 11 voters. This gives us a rough but intriguing way to compare the relative strengths of participation and half-way monotonicity (see Table~\ref{tbl:results}); we can see that half-way monotonicity is weaker than participation, but not by much.

In Section~\ref{sec:extensions}, we will discuss some extensions of this result. First, we consider \emph{irresolute} voting rules which return a \emph{set} of alternatives; we check whether in this more general model we can guarantee half-way monotonicity for larger numbers of voters (the answer turns out to be \emph{no}). Then we consider the \emph{strong} preference reversal paradox, which occurs when a voter can cause their \emph{most}-preferred alternative to win by reversing their preferences. We show that most, but not all, Condorcet extensions exhibit this strong paradox.

Finally, in Section~\ref{sec:zwicker}, we combine our results with a theorem of Campbell and Kelly~\cite{CaKe03a} to give a strengthened version of the Gibbard--Satterthwaite Theorem. This version claims that every non-trivial voting rule is either \emph{needlessly} or \emph{egregiously} manipulable. This gives a more explicit description of the types of manipulations that are sufficient to obtain an impossibility.

\section{Definitions}

A \emph{linear order} $\pref$ is a complete, antisymmetric, transitive binary relation over $A$. We write $\succ$ for the strict (irreflexive) part of $\pref$. The set of all linear orders over $A$ is denoted by $A!$. The \emph{reverse} $\succ\rev$ of a linear order $\succ$ is defined by $a \succ\rev b \iff b \succ a$ for all $a,b \in A$.

Let $N = \{1,\dots,n\}$ be a finite set of $n$ \emph{voters}, and let $A$ be a finite set of $m$ \emph{alternatives}. Often, we will consider the case of precisely 4 alternatives, when $A = \{a,b,c,d\}$. A \emph{profile} $P$ is a function assigning to every $i\in N$ a linear ordering $\pref_i$ of the alternatives. Thus, the set of profiles is $A!^N$. A (resolute) \emph{voting rule} is a function $f : A!^N \to A$ that assigns a winning alternative $f(P) \in A$ to every profile $P \in A!^N$.

Given a profile $P$, we say that $a \in A$ is the (unique) \emph{Condorcet winner} if $|\{ i \in N : a \succ_i b \}| > |\{ i \in N : b \succ_i a \}|$ for all $b\in A \setminus \{a\}$. Thus, a Condorcet winner wins against every other alternative in a pairwise majority comparison. We say that a voting rule $f$ is a \emph{Condorcet extension} if $f$ selects the Condorcet winner whenever one exists.

Given a profile $P \in A!^N$ where $i\in N$, we write $P_{-i} := P |_{N\setminus\{i\}}$ for the profile obtained from $P$ by removing voter $i$. We also write $(P_{-i}, \succ_i') := P_{-i} \cup \{ (i, \succ_i') \}$ for the profile obtained from $P$ by replacing $i$'s vote by $\succ_i'$.

\begin{definition}
	A voting rule $f$ satisfies \emph{half-way monotonicity} if
	\[ f(P_{-i}, \succ_i) \pref_i f(P_{-i}, \succ_i\rev)\quad \text{for all profiles $P \in A!^N$ and all voters $i\in N$.} \]
	Thus, voters weakly prefer voting truthfully to voting the reverse of their preferences. If a rule violates half-way monotonicity, we say that it suffers from the \emph{preference reversal paradox.}
\end{definition}

\section{Relationship to Participation}

\emph{Participation} is a property of voting rules that assign outcomes to profiles with varying numbers of voters. Let us define a \emph{variable-electorate} voting rule as a function that assigns a winning alternative to every profile defined on some finite electorate $N \subseteq \mathbb{N}$. If $N$ is an electorate with $i\not\in N$, $\succ_i$ is some linear order, and $P$ is a profile on $N$, then we define $P + (i, \succ_i)$ to be the profile obtained by letting voter $i$ join $P$. Then we say that a variable-electorate voting rule $f$ satisfies \emph{participation} if for all electorates $N$, all voters $i\not\in N$, and all preference orders $\succ_i$, we have $f(P + (i, \succ_i)) \pref_i f(P)$. In other words, voters always weakly prefer joining an election.

It turns out that participation is a stronger requirement than half-way monotonicity. This was shown by Sanver and Zwicker~\cite[Theorem 4.1]{SaZw09a} using a proof that established several interrelated implications among their monotonicity axioms. Here, we give a direct proof of this implication.

\begin{lemma}[Sanver and Zwicker, \citealp{SaZw09a}]
	\label{lem:part-implies-hwm}
	If a variable-electorate voting rule $f$ satisfies participation, then $f$ satisfies half-way monotonicity.
\end{lemma}
\begin{proof}
	The key idea is that the reversal of a vote $\succ_i$ is equivalent to $\succ_i$ leaving the election and $\succ_i\rev$ joining it. Let $P \in A!^N$ be a profile and let $i\in N$ be a voter with preferences $\pref_i$ in $P$. Consider the profile $P_{-i}$ with $i$ removed. By participation, we have $f(P) \pref_i f(P_{-i})$. Also by participation, we have $f(P_{-i}, \succ_i\rev) \pref_i\rev f(P_{-i})$. Putting these together, using the definition of the reverse of an order, we have
	\[ f(P) \pref_i f(P_{-i}) \pref_i f(P_{-i}, \succ_i\rev). \]
	Thus, using transitivity, we have verified half-way monotonicity.
\end{proof}

Interestingly, to deduce half-way monotonicity for electorates of $n$ voters, we only require participation to hold between electorates of size $n-1$ and $n$. N\'u\~nez and Sanver~\cite{NuSa17a} also prove the implication of Lemma~\ref{lem:part-implies-hwm} by proposing an intermediate ``Condition $\lambda$'' that is implied by participation and that implies half-way monotonicity.

\section{Method}
\label{sec:method}

To obtain the possibility and impossibility results of the next section, we used the computer-aided technique introduced by Geist and Endriss~\cite{GeEn11a} and Tang and Lin~\cite{TaLi09a}. In this section, we will give a brief overview of the basic ideas. For a more detailed discussion of the method, see the survey by Geist and Peters~\cite{GePe17a}.

We begin by translating our question (of whether a Condorcet extension satisfying half-way monotonicity exists) into \emph{propositional logic}. To do so, we fix a set $A$ of $m$ alternatives and a set $N$ of $n$ voters. We then explicitly enumerate the set $A!^N$ of profiles, and introduce propositional variables $x_{P,a}$ for each profile $P \in A!^N$ and each alternative $a\in A$. The intended meaning of the variables is
\[ x_{P,a} \text{ is set true} \iff f(P) = a,  \]
where $f$ is a voting rule. To pin down this meaning, we produce a propositional formula $\varphi$ in conjunctive normal form (CNF) by encoding three classes of constraints:
\begin{itemize}
	\item \emph{functionality} of $f$, i.e., that $f(P) = a$ for exactly one alternative, which means that there is \emph{at least} one and \emph{at most} one such alternative:
	\begin{align*}
	\varphi_{\text{functionality}} &\equiv \bigwedge_{P\in A!^N} \left( \left( \bigvee_{a\in A} x_{P,a} \right) \land \bigwedge_{a\neq b\in A} (\lnot x_{P,a} \lor \lnot x_{P,b}) \right)
	\intertext{\item \emph{Condorcet-consistency}: for $a\in A$, let $C_a \subseteq A!^N$ be the set of profiles in which $a$ is the Condorcet winner.}
	\varphi_{\text{Condorcet}} &\equiv \bigwedge_{a\in A} \bigwedge_{P\in C_a} x_{P,a}
	\intertext{\item \emph{half-way monotonicity}: if $a,b\in A$ are such that $a\succ_i b$ for voter $i$ in profile $P \in A!^N$, and $f(P) = b$, then $f(P_{-i}, \succ_i\rev) \neq a$.}
	\varphi_{\text{half-way monotonicity}} &\equiv \bigwedge_{i\in N} \bigwedge_{P \in A!^N}  \bigwedge_{\substack{a,b\in A \\ a \succ_i b}} (\lnot v_{P,b} \lor \lnot v_{(P_{-i}, \succ_i\rev),a}). 
	\end{align*}
\end{itemize}

Putting these formulas together, we obtain $\varphi \equiv \varphi_{\text{functionality}} \land \varphi_{\text{Condorcet}} \land \varphi_{\text{half-way monotonicity}}$. Then it is clear that each true/false assignment to the propositional variables that satisfies $\varphi$ induces a voting rule $f$ which is Condorcet-consistent and satisfies half-way monotonicity.

Next, we write down $\varphi$ in a text file in the standardised DIMACS format, and pass this formula to a \emph{SAT solver}, that is, a computer program which checks whether a given propositional formula is satisfiable or unsatisfiable. Despite this decision problem being NP-complete, modern SAT solvers such as lingeling \citep{Bier13a} or glucose \citep{AuSi09a} can often solve even large formulas in a relatively short time. 

For our choice of $\varphi$, a satisfiability result gives us an example of a Condorcet extension which avoids the preference reversal paradox. An unsatisfiability result implies an impossibility. In the former case, the SAT solver will return a satisfying assignment, which induces a specific voting rule $f$ which is available as a look-up table. In the latter case, the SAT solver will merely report ``UNSAT''. It would be desirable to obtain a proof of this claim. While many solvers are able to produce an unsatisfiability proof in a computer-readable format, these proofs can be very big (a recent result in Ramsey theory required 200 TB \citep{HKM16a}) and uninsightful. Following Brandt and Geist~\cite{BrGe15a}, we use a technique that is sometimes able to produce short and human-readable proofs. To do so, we obtain a \emph{minimal unsatisfiable set} (MUS), which is a subselection of the clauses of $\varphi$ which is already unsatisfiable. If we can find a small enough MUS using tools such as MUSer2 \citep{BeMa12a} or MARCO \citep{LPMM15a}, we can then translate this MUS into a human-readable proof. 

For more details on this approach, we refer to Geist and Peters~\cite{GePe17a} and to Brandt, Geist, and Peters~\cite{BGP16c}.

\section{Impossibility Results}

In this section, we will present our main results. Our technique is inductive: we give positive and negative results for a specific number $n$ of voters, and then use the following lemma to conclude that positive results also hold for smaller $n$ and negative results hold for larger $n$, as long as parity is preserved. While the parity of $n$ is immaterial as to whether participation can be satisfied by a variable-electorate voting rule defined on electorates up to size $n$, we will see that half-way monotonicity is less restrictive on Condorcet extensions defined for even electorates.

\begin{lemma}[Induction Step]
	\label{lem:induction}
	Fix a number $m$ of alternatives, and let $n\ge 1$. If there exists a Condorcet extension defined on electorates with $n+2$ voters which satisfies half-way monotonicity, then there also exists a Condorcet extension for $n$ voters satisfying half-way monotonicity.
\end{lemma}
\begin{proof}
	Fix some linear order $\succ_*$ over $A$. Suppose $|N| = n$, and suppose $f_{n+2}$ is a Condorcet extension satisfying half-way monotonicity, defined for the electorate $N \cup \{i,j\}$. Then define the voting rule $f_n$ on the electorate $N$ by
	\[ f_n(P) := f_{n+2}(P + (i, \succ_*) + (j, \succ_*\rev)) \text{ for all profiles $P \in A!^N$}. \]
	Then the voting rule $f_n$ is Condorcet-consistent: if a profile $P\in A!^N$ admits a Condorcet winner, then this alternative remains the Condorcet winner after adding two completely opposed orders to $P$, since this operation does not change the majority margins. Further, $f_n$ satisfies half-way monotonicity, since any successful manipulation through preference reversal for $f_n$ can also be pulled off for $f_{n+2}$.
\end{proof}
Contrapositively, this lemma implies that an incompatibility result between Condorcet-consistency and half-way monotonicity for $n$ voters also applies to $n + 2k$ voters, for each $k \ge 0$. Thus, in our impossibility proofs below, we only need to handle the base case for $n = 15$ and $n = 24$, respectively.

Before we present the proofs, let us have a look at our main positive result.

\begin{proposition}[Possibilities]
	\label{prop:positive}
	For $m = 4$ alternatives, and for either $n = 13$ or $n = 22$ voters, there exists a Condorcet extension satisfying half-way monotonicity.
\end{proposition}
This result was obtained by running a SAT solver to decide the satisfiability of a suitable (large) formula of propositional logic as described in Section~\ref{sec:method}. The major downside of this technique is that the voting rules of Proposition~\ref{prop:positive} are only available as \emph{look-up tables}. Both of the voting rules mentioned are C2 functions in Fishburn's classification, i.e., they only depend on the majority margins of the input profile.\footnote{There will also exist other example functions that satisfy our axioms but are not C2; restricting attention to C2 functions allows our computer search approach to be tractable. We do not have an explanation for why this restriction still allows for tight bounds.} 
The only available description of these voting rules are text files indicating, for every weighted majority tournament, which alternative is to be selected. For the case of $n = 22$ voters, the uncompressed file has a size of 1.7GB. As with the voting rules found by Brandt, Geist, and Peters~\cite{BGP16c}, it would be desirable to find such rules that have a more concise description.

Now let us move on to our negative results. The proof diagrams in Figures~\ref{fig:odd} and~\ref{fig:even} give a graphical representation of the proof steps involved. 
An arc from $P$ to $P'$ labelled ``\emph{reverse 2 $\no{\c\a}\yes{\b\d}$}'' is interpreted as ``if the voting rule chooses $a$ or $b$ at $P$, then the rule must also choose $a$ or $b$ at $P'$ by half-way monotonicity''. The profiles at the leafs all admit a Condorcet winner, which leads to a contradiction. 
The general proof strategy of our impossibility proofs is a follows: we identify an initial profile $P_0$, and iterate through each possible value of $f(P_0) \in A$. Assuming that $f(P_0) = x$, say, will then, by half-way monotonicity, imply restrictions on the possible values that $f$ can take at profiles obtained from $P_0$ by reversing some of the votes. In particular, it will imply that $f$ must not pick the Condorcet winner at some of these profiles, contradicting $f$ being a Condorcet extension.

As we noted, we will treat the cases of odd and even electorates separately, since the induction step of Lemma~\ref{lem:induction} only works in steps of two. Let us start with the odd case.

\begin{figure}[p]
	\centering
	\begin{tikzpicture}
	[->,
	level distance=23mm,
	sibling distance=55mm,
	level 2/.style={level distance=21mm,sibling distance=24mm},
	level 3/.style={level distance=21mm,sibling distance=13mm},
	ne/.style={inner sep=2pt},
	empty/.style={circle,draw=black!75,fill=black!40, inner sep=1.5pt},
	condorcet/.style={},
	lbl/.style={fill=white, inner sep=3pt, align=center}
	]
	
	\begin{scope}
	[xscale=1.57]
	\draw node [inner sep=5pt] 
	(fake root) at (0, 0) 
	{
		\begin{array}[t]{c@{\quad}c@{\quad}c@{\quad}c@{\quad}c@{\quad}c}
		\toprule
		1&3&3&4&2&2\\
		\midrule
		\a&\a&\b&\c&\d&\d\\
		\b&\b&\d&\a&\c&\c\\
		\c&\d&\c&\b&\a&\b\\
		\d&\c&\a&\d&\b&\a\\
		\bottomrule\\[-5pt]
		\multicolumn{6}{c}{P_0}
		\end{array}
	};
	\draw node [rectangle, minimum width=5.5cm, minimum height=1.9cm] 
	(root) at (0, 1cm) 
	{}
	child {
		node (alpha) {$P_1$}
		child {
			node[empty] (1) {}
			child {
				node[condorcet] (ca) {$P_3$}
				edge from parent node[lbl] {reverse \\ 2 $\no{\c\a}\yes{\b\d}$}
			}
			edge from parent node[lbl] {reverse \\ 1 $\no{\d\c\a}\yes{\b}$}
		}
		child {
			node[condorcet] (cc) {$P_2$}
			edge from parent node[lbl] {reverse \\ 2 $\no{\b\d\c}\yes{\a}$}
		}
		edge from parent 
		node[lbl] {reverse \\ 1 $\no{\d\c}\yes{\b\a}$}
	}
	child {
		node (alphaprime) {$P_4$}
		child {
			node[condorcet] (cbprime) {$P_5$}
			edge from parent node[lbl] {reverse \\ 2 $\no{\c\a\b}\yes{\d}$}
		}
		child {
			node[condorcet] (cdprime) {$P_6$}
			edge from parent node[lbl] {reverse \\ 3 $\no{\a\b\d}\yes{\c}$}
		}
		edge from parent node[lbl] {reverse \\ 1 $\no{\a\b}\yes{\c\d}$}
	}
	;
	\end{scope}
	\begin{scope}[shift={(ca.south)},shift={(0cm,-1cm)}]
	\draw node [inner sep=5pt, scale=0.7] (profile)
	{
		\begin{array}[t]{c@{\quad}c@{\quad}c@{\quad}c@{\quad}c@{\quad}c@{\quad}c@{\quad}c}
		\toprule
		2&3&1&3&2&2&1&1\\
		\midrule
		\yes\a&\yes\a&\b&\b&\c&\d&\d&\d\\
		\b&\b&\yes\a&\d&\yes\a&\b&\c&\c\\
		\c&\d&\c&\c&\b&\yes\a&\yes\a&\b\\
		\d&\c&\d&\yes\a&\d&\c&\b&\yes\a\\
		\bottomrule
		\end{array}
	};
	\end{scope}
	\begin{scope}[shift={(cc.south)},shift={(0cm,-1cm)}]
	\draw node [inner sep=5pt, scale=0.7] (profile)
	{
		\begin{array}[t]{c@{\quad}c@{\quad}c@{\quad}c@{\quad}c@{\quad}c@{\quad}c}
		\toprule
		2&3&2&1&4&2&1\\
		\midrule
		\a&\a&\a&\b&\yes\c&\d&\d\\
		\b&\b&\yes\c&\d&\a&\yes\c&\yes\c\\
		\yes\c&\d&\d&\yes\c&\b&\a&\b\\
		\d&\yes\c&\b&\a&\d&\b&\a\\
		\bottomrule
		\end{array}
	};
	\end{scope}
	\begin{scope}[shift={(cbprime.south)},shift={(0cm,-1cm)}]
	\draw node [inner sep=5pt, scale=0.7] (profile)
	{
		\begin{array}[t]{c@{\quad}c@{\quad}c@{\quad}c@{\quad}c@{\quad}c}
		\toprule
		3&3&2&2&2&3\\
		\midrule
		\a&\yes\b&\c&\d&\d&\d\\
		\yes\b&\d&\a&\yes\b&\c&\c\\
		\d&\c&\yes\b&\a&\a&\yes\b\\
		\c&\a&\d&\c&\yes\b&\a\\
		\bottomrule
		\end{array}
	};
	\end{scope}
	
	\begin{scope}[shift={(cdprime.south)},shift={(0cm,-1cm)}]
	\draw node [inner sep=5pt, scale=0.7] (profile)
	{	\quad
		\begin{array}[t]{c@{\quad}c@{\quad}c@{\quad}c@{\quad}c}
		\toprule
		3&4&3&2&3\\
		\midrule
		\b&\c&\c&\yes\d&\yes\d\\
		\yes\d&\a&\yes\d&\c&\c\\
		\c&\b&\b&\a&\b\\
		\a&\yes\d&\a&\b&\a\\
		\bottomrule
		\end{array}
		\qquad
	};
	\end{scope}
	\end{tikzpicture}
	\caption{Proof diagram of the proof of Theorem~\ref{thm:odd}.}
	\label{fig:odd}
\end{figure}
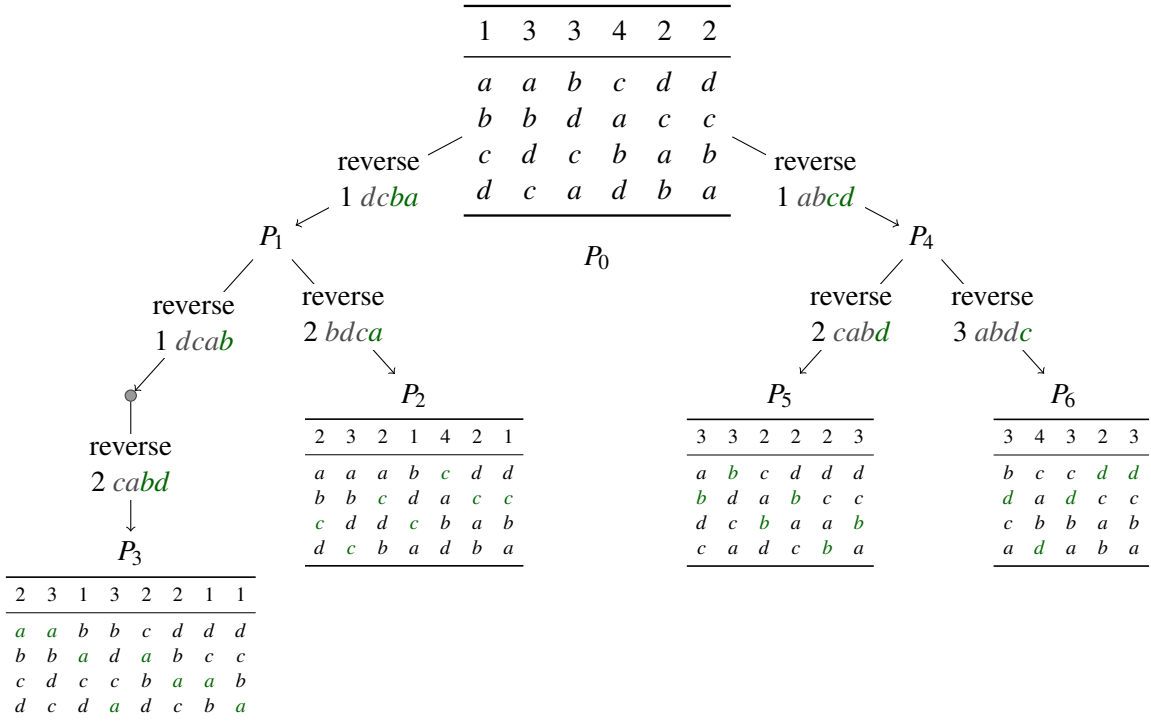

\begin{figure}[p]
	\centering
	\begin{tikzpicture}
	[->,
	level distance=23mm,
	sibling distance=55mm,
	level 2/.style={level distance=21mm,sibling distance=24mm},
	level 3/.style={level distance=21mm,sibling distance=13mm},
	ne/.style={inner sep=2pt},
	empty/.style={circle,draw=black!75,fill=black!40, inner sep=1.5pt},
	condorcet/.style={},
	lbl/.style={fill=white, inner sep=3pt, align=center}
	]
	
	\begin{scope}
	[xscale=1.57]
	\draw node [inner sep=5pt] 
	(fake root) at (0, 0) 
	{
		\begin{array}[t]{c@{\quad}c@{\quad}c@{\quad}c@{\quad}c@{\quad}c}
		\toprule
		2&4&6&6&4&2\\
		\midrule
		\a&\a&\b&\c&\d&\d\\
		\b&\b&\d&\a&\c&\c\\
		\c&\d&\c&\b&\a&\b\\
		\d&\c&\a&\d&\b&\a\\
		\bottomrule\\[-5pt]
		\multicolumn{6}{c}{P_0}
		\end{array}
	};
	\draw node [rectangle, minimum width=5cm, minimum height=1.9cm] 
	(root) at (0, 1cm) 
	{}
	child {
		node (alpha) {$P_1$}
		child {
			node[empty] (1) {}
			child {
				node[condorcet] (ca) {$P_3$}
				edge from parent node[lbl] {reverse \\ 3 $\no{\c\a}\yes{\b\d}$}
			}
			edge from parent node[lbl] {reverse \\ 2 $\no{\d\c\a}\yes{\b}$}
		}
		child {
			node[condorcet] (cc) {$P_2$}
			edge from parent node[lbl] {reverse \\ 3 $\no{\b\d\c}\yes{\a}$}
		}
		edge from parent 
		node[lbl] {reverse \\ 2 $\no{\d\c}\yes{\b\a}$}
	}
	child {
		node (alphaprime) {$P_4$}
		child {
			node[condorcet] (cbprime) {$P_5$}
			edge from parent node[lbl] {reverse \\ 3 $\no{\c\a\b}\yes{\d}$}
		}
		child {
			node[empty] (2) {}
			child {
				node[condorcet] (cdprime) {$P_6$}
				edge from parent node[lbl] {reverse \\ 3 $\no{\b\d}\yes{\c\a}$}
			}
			edge from parent node[lbl] {reverse \\ 2 $\no{\a\b\d}\yes{\c}$}
		}
		edge from parent node[lbl] {reverse \\ 2 $\no{\a\b}\yes{\c\d}$}
	}
	;
	\end{scope}
	\begin{scope}[shift={(ca.south)},shift={(0cm,-1cm)}]
	\draw node [inner sep=5pt, scale=0.7] (profile)
	{
		\begin{array}[t]{c@{\quad}c@{\quad}c@{\quad}c@{\quad}c@{\quad}c@{\quad}c}
		\toprule
		4&4&2&6&3&3&2\\
		\midrule
		\yes\a&\yes\a&\b&\b&\c&\d&\d\\
		\b&\b&\yes\a&\d&\yes\a&\b&\c\\
		\c&\d&\c&\c&\b&\yes\a&\yes\a\\
		\d&\c&\d&\yes\a&\d&\c&\b\\
		\bottomrule
		\end{array}
	};
	\end{scope}
	\begin{scope}[shift={(cc.south)},shift={(0cm,-1cm)}]
	\draw node [inner sep=5pt, scale=0.7] (profile)
	{
		\begin{array}[t]{c@{\quad}c@{\quad}c@{\quad}c@{\quad}c@{\quad}c}
		\toprule
		4&4&3&3&6&4\\
		\midrule
		\a&\a&\a&\b&\yes\c&\d\\
		\b&\b&\yes\c&\d&\a&\yes\c\\
		\yes\c&\d&\d&\yes\c&\b&\a\\
		\d&\yes\c&\b&\a&\d&\b\\
		\bottomrule
		\end{array}
	};
	\end{scope}
	\begin{scope}[shift={(cbprime.south)},shift={(0cm,-1cm)}]
	\draw node [inner sep=5pt, scale=0.7] (profile)
	{
		\begin{array}[t]{c@{\quad}c@{\quad}c@{\quad}c@{\quad}c@{\quad}c}
		\toprule
		4&6&3&3&4&4\\
		\midrule
		\a&\b&\c&\d&\d&\d\\
		\b&\d&\a&\b&\c&\c\\
		\d&\c&\b&\a&\a&\b\\
		\c&\a&\d&\c&\b&\a\\
		\bottomrule
		\end{array}
	};
	\end{scope}
	
	\begin{scope}[shift={(cdprime.south)},shift={(0cm,-1cm)}]
	\draw node [inner sep=5pt, scale=0.7] (profile)
	{
		\begin{array}[t]{c@{\quad}c@{\quad}c@{\quad}c@{\quad}c@{\quad}c@{\quad}c}
		\toprule
		2&3&3&6&2&4&4\\
		\midrule
		\a&\a&\b&\c&\c&\d&\d\\
		\b&\c&\d&\a&\d&\c&\c\\
		\d&\d&\c&\b&\b&\a&\b\\
		\c&\b&\a&\d&\a&\b&\a\\
		\bottomrule
		\end{array}
	};
	\end{scope}
	\end{tikzpicture}
	\caption{Proof diagram of the proof of Theorem~\ref{thm:even}.}
	\label{fig:even}
\end{figure}
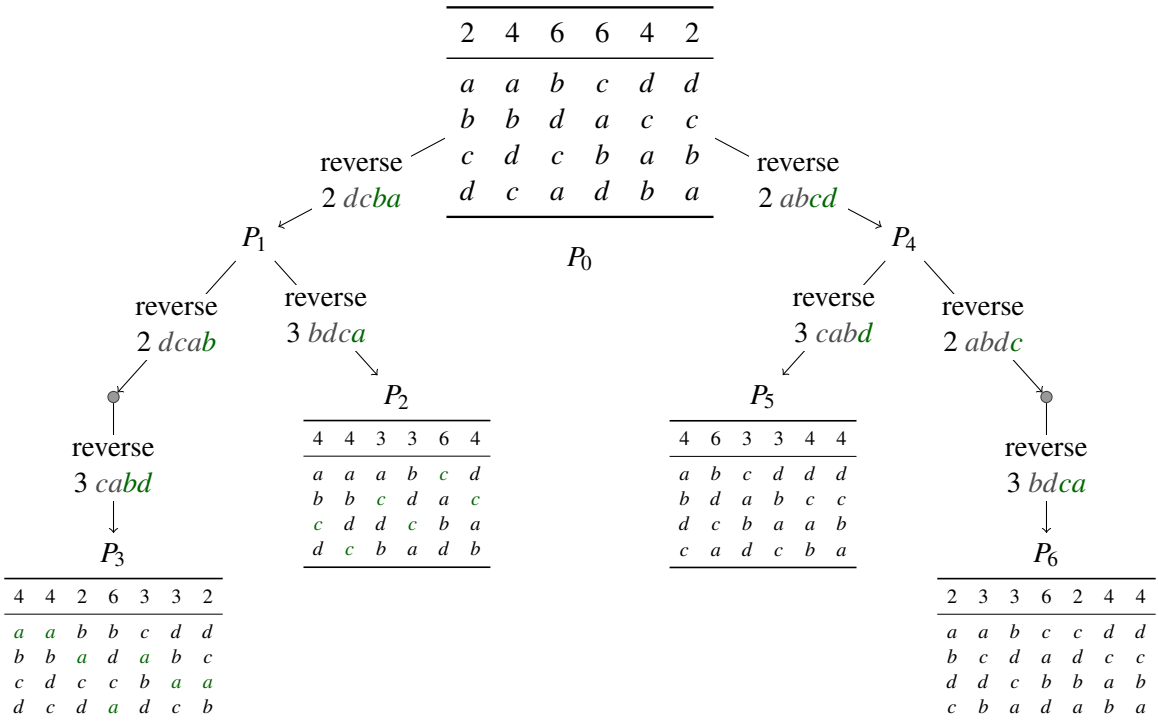

\begin{theorem}[Odd Electorates]
	\label{thm:odd}
	For $m\ge 4$ alternatives and odd $n\ge 15$, there does not exist a Condorcet extension satisfying half-way monotonicity.
\end{theorem}

\setlength{\columnsep}{7pt}
\setlength{\intextsep}{1pt}
\begin{wraptable}[8]{r}{0.26\linewidth}
	\begin{tabular}{cccccc}
		\toprule
		1&3&3&4&2&2\\
		\midrule
		\a&\a&\b&\c&\d&\d\\
		\b&\b&\d&\a&\c&\c\\
		\c&\d&\c&\b&\a&\b\\
		\d&\c&\a&\d&\b&\a\\
		$X$&$X$&$X$&$X$&$X$&$X$\\
		\bottomrule 
	\end{tabular}
\end{wraptable}
\noindent
\emph{Proof} By Lemma~\ref{lem:induction}, we only need to handle the case with $n = 15$. Write $A = \{a,b,c,d\} \cup X$, where $X = \{x_1,\dots,x_{m-4}\}$. Suppose there exists a half-way monotonic Condorcet extension $f$ for 15 voters. Consider the 15-voter profile $P_0$ depicted on the right. The column numbers indicate how many voters submit a given ordering; for example, there are exactly 3 voters in $P_0$ with the ordering $a \succ b \succ d \succ c \succ X$; let us abbreviate this ordering as ``$abdc$'', and similarly for other voters. The $X$ at the bottom of each vote should be replaced by an arbitrary ordering of the alternatives in $X$. Our proof is by case analysis on the value of $f(P_0)$, arriving at a contradiction in each case.

Suppose first that $f(P_0) \in \{a,b\} \cup X$. Let $P_1$ be the profile after one $dcba$ voter reverses their preferences in $P_0$. By half-way monotonicity, we have $f(P_1) \in \{a,b\} \cup X$. Suppose that $f(P_1) \in \{a\} \cup X$. Let $P_2$ be the profile after two $bdca$ voters reverse their preferences in $P_1$. By half-way monotonicity, we have $f(P_2) \in \{a\} \cup X$; however $c$ is the Condorcet winner in $P_2$, contradicting Condorcet-consistency of $f$. Thus $f(P_1)  = b$. Let $P_3$ be the profile obtained from $P_1$ after one $dcab$ voter and two $cabd$ voters reverse their preferences. By half-way monotonicity, we have $f(P_3) \in \{b,d\}$. However, $a$ is the Condorcet winner in $P_3$, a contradiction.

Thus $f(P_0) \in \{c,d\}$. Let $P_4$ be the profile obtained from $P_0$ by reversing an $abcd$ voter. By half-way monotonicity, $f(P_4) \in \{c,d\}$. Suppose $f(P_4) = d$. Let $P_5$ be the profile obtained from $P_4$ by reversing two $cabd$ voters; then $f(P_5) = d$. But $b$ is the Condorcet winner at $P_5$, a contradiction. Hence $f(P_4) = c$. Let $P_6$ be the profile obtained from $P_4$ by reversing three $abdc$ voters; then $f(P_6) = c$. But $d$ is the Condorcet winner at $P_6$, a contradiction. 
\qed

\medskip
\noindent
The proof above was obtained with the help of computers, and in particular by using SAT solvers, in a way similar to the technique described by Brandt, Geist, and Peters~\cite{BGP16c}. In particular, our search was aided by only considering profiles made up of the about 6--10 preference orders that appear in their proofs of the no-show paradox.

The bound on $n$ for even electorates is significantly higher than for odd ones. Intuitively, the reason is that Condorcet-consistency is less demanding in even electorates, since there are `fewer' Condorcet winners because they need to beat every other alternative by a majority margin of at least 2.

\begin{theorem}[Even Electorates]
	\label{thm:even}
	For $m\ge 4$ alternatives and even $n\ge 24$, there does not exist a Condorcet extension satisfying half-way monotonicity.
\end{theorem}

\setlength{\columnsep}{7pt}
\setlength{\intextsep}{1pt}
\begin{wraptable}[8]{r}{0.26\linewidth}
	\begin{tabular}{cccccc}
		\toprule
		2&4&6&6&4&2\\
		\midrule
		\a&\a&\b&\c&\d&\d\\
		\b&\b&\d&\a&\c&\c\\
		\c&\d&\c&\b&\a&\b\\
		\d&\c&\a&\d&\b&\a\\
		$X$&$X$&$X$&$X$&$X$&$X$\\
		\bottomrule 
	\end{tabular}
\end{wraptable}
\noindent
\emph{Proof} By Lemma~\ref{lem:induction}, we only need to handle the case with $n = 24$. Write $A = \{a,b,c,d\} \cup X$, where $X = \{x_1,\dots,x_{m-4}\}$. Suppose there exists a half-way monotonic Condorcet extension $f$ for 24 voters. Consider the 24-voter profile $P_0$ depicted on the right. The column numbers indicate how many voters submit a given ordering; for example, there are exactly 4 voters in $P_0$ with the ordering $a \succ b \succ d \succ c \succ X$. The $X$ at the bottom should be replaced by an arbitrary ordering of the alternatives in $X$. Our proof is by case analysis on the value of $f(P_0)$, arriving at a contradiction in each case.

Suppose first that $f(P_0) \in \{a,b\} \cup X$. Let $P_1$ be the profile after two $dcba$ voters reverses their preferences in $P_0$. By half-way monotonicity, we have $f(P_1) \in \{a,b\} \cup X$. Suppose that $f(P_1) \in \{a\} \cup X$. Let $P_2$ be the profile after three $bdca$ voters reverse their preferences in $P_1$. By half-way monotonicity, we have $f(P_2) \in \{a\} \cup X$; however $c$ is the Condorcet winner in $P_2$, contradicting Condorcet-consistency of $f$. Thus $f(P_1)  = b$. Let $P_3$ be the profile obtained from $P_1$ after two $dcab$ voter and three $cabd$ voters reverse their preferences. By half-way monotonicity, we have $f(P_3) \in \{b,d\}$. However, $a$ is the Condorcet winner in $P_3$, a contradiction.

Thus $f(P_0) \in \{c,d\}$. Let $P_4$ be the profile obtained from $P_0$ by reversing two $abcd$ voters. By half-way monotonicity, $f(P_4) \in \{c,d\}$. Suppose $f(P_4) = d$. Let $P_5$ be the profile obtained from $P_4$ by reversing three $cabd$ voters; then $f(P_5) = d$. But $b$ is the Condorcet winner at $P_5$, a contradiction. Hence $f(P_4) = c$. Let $P_6$ be the profile obtained from $P_4$ by reversing two $abdc$ and three $bdca$ voters; then $f(P_6) = c$. But $d$ is the Condorcet winner at $P_6$, a contradiction.
\qed

\medskip
\noindent
One may wonder whether it is a coincidence that our cut-off for half-way monotonicity in even electorates ($n=24$) is double the cut-off for participation ($n=12$). The answer is no, as suggested by the proof of Theorem~4.1(3) of Sanver and Zwicker~\cite{SaZw09a}, which (roughly) shows that half-way monotonicity for $2n$ voters implies participation for $n$ voters, at least in the presence of homogeneity and reversal cancellation. In fact, we have obtained the proof of Theorem~\ref{thm:even} by taking an impossibility proof for the no-show paradox for $n=12$, and doubling all the profiles involved in the proof.

\section{Extensions: Irresolute Rules and Strong Paradoxes}
\label{sec:extensions}

\subsection{Irresolute Voting Rules}
\label{sec:irresolute}

A \emph{set-valued} (or \emph{irresolute}) voting rule is a function $F : A!^N \to 2^A \setminus \{ \emptyset \}$ that assigns to every profile~$P$ a non-empty subset $F(P) \subseteq A$ of winning alternatives. The usual interpretation is that the ties will later be broken by some other mechanism. A set-valued voting rule $F$ is a \emph{Condorcet extension} if it uniquely selects the Condorcet winner if one exists; thus, if $x$ is the Condorcet winner of the profile $P$, then $F(P) = \{x\}$. One can define analogues of half-way monotonicity for set-valued voting rules in several ways \citep[see Sanver and Zwicker,][]{SaZw10a}. Here, we will focus on the approach using \emph{set extensions}, where voters' preferences over alternatives are lifted to preferences over \emph{sets} of alternatives. In particular, following Jimeno et al.~\cite{JPG09a}, we focus on the \emph{optimistic} and the \emph{pessimistic} set extensions that are the subject of the Duggan--Schwartz Impossibility Theorem \cite{DuSc00a}. An optimist prefers sets with better most-preferred alternative, while a pessimist prefers sets with better least-preferred alternative. If $X = \{a,d\}$ and $Y = \{b,c\}$, then an optimist with preferences $a \succ b \succ c \succ d$ would prefer $X$ to $Y$, while a pessimist with the same underlying preferences would prefer $Y$ to $X$.

Given a set $X\subseteq A$, let us write $\max_{\pref_i} X$ (resp. $\min_{\pref_i} X$) for the most-preferred (resp. least-preferred) alternative in $X$ according to $\pref_i$, so that for all $x\in X$ we have $\max_{\pref_i} X \pref_i x \pref_i \min_{\pref_i} X$. This allows us to define variants of half-way monotonicity for these set extensions:

\begin{definition}
	A set-valued voting rule $F$ satisfies \emph{optimistic half-way monotonicity} if
	\[ \max_{\pref_i} F(P_{-i}, \succ_i) \pref_i \max_{\pref_i} F(P_{-i}, \succ_i\rev) \quad \text{for all profiles $P$ and all voters $i$.} \]
	A set-valued voting rule $F$ satisfies \emph{pessimistic half-way monotonicity} if
	\[ \min_{\pref_i} F(P_{-i}, \succ_i) \pref_i \min_{\pref_i} F(P_{-i}, \succ_i\rev) \quad \text{for all profiles $P$ and all voters $i$.} \]
\end{definition}

One might hope that dropping the requirement of resoluteness makes it easier for Condorcet extensions to satisfy half-way monotonicity. Indeed, this happens for participation: As Brandt, Geist, and Peters~\cite{BGP16c} show, there are Condorcet extensions satisfying optimistic participation for 16 voters and pessimistic participation for 13 voters, while the limit is 11 voters for resolute rules. For half-way monotonicity, surprisingly, it turns out that giving up resoluteness buys us nothing:

\begin{theorem}
	\label{thm:irresolute}
	For $m\ge 4$ alternatives and odd $n\ge 15$ or even $n\ge 24$, there does not exist a set-valued Condorcet extension satisfying either pessimistic or optimistic half-way monotonicity.
\end{theorem}

\noindent
Why does the move to the irresolute setting not allow us larger bounds on $n$, when it does for participation? An intuitive reason is suggested by the proof of Lemma~\ref{lem:part-implies-hwm}, where we showed that participation implies half-way monotonicity by decomposing a preference reversal into a voter \emph{leaving} the electorate and the reverse voter \emph{joining} the electorate. Now, a \emph{pessimist} reversing their preferences can be decomposed into the pessimist leaving, and an optimist with reverse preferences joining. Thus, pessimistic half-way monotonicity is related to the \emph{conjunction} of optimistic and pessimistic participation, and neither of the latter properties alone implies pessimistic half-way monotonicity. As Brandt, Geist, and Peters~\cite{BGP16c} find, imposing this conjunction of properties does not allow for larger bounds in the participation setting either.

\begin{proof}[Proof of Theorem~\ref{thm:irresolute}]
	For pessimistic half-way monotonicity, we can follow the proofs of Theorems~\ref{thm:odd} and~\ref{thm:even} almost verbatim. By way of example, let us translate the second paragraph of the proof of Theorem~\ref{thm:odd}. Let $P_0,\dots,P_6$ refer to the same profiles as in that proof. The end result of the proof is to conclude that $F(P_0) = \emptyset$, a contradiction. Suppose that $F(P_0)$ intersects $\{a,b\}\cup X$. Then in $P_1$, where one $dcba$ voter is reversed, we also have $F(P_1)$ intersecting $\{a,b\}\cup X$, by pessimistic half-way monotonicity (since the minimum of the $dcba$ voter must have gone weakly down). Suppose in fact that $F(P_1)$ intersects $\{a\}\cup X$. Then in $P_2$, after reversing two $bdca$ voters, we again have that $F(P_2)$ intersects $\{a\}\cup X$, by pessimistic half-way monotonicity. But since $c$ is the Condorcet winner in $P_2$, we have $F(P_2) = \{c\}$, a contradiction. Hence $F(P_1)$ must intersect $\{b\}$, i.e., $b\in F(P_1)$. But then, similarly, $F(P_3)$ must intersect $\{b,d\}$, contradicting $F(P_3) = \{a\}$ by Condorcet-consistency. So $F(P_0)$ cannot intersect $\{a,b\} \cup X$ after all, hence must intersect $\{c,d\}$. Following the proof steps about $P_4,P_5,P_6$, we see that this also leads to contradiction.
	
	For optimistic half-way monotonicity, we work through the proofs ``the other way around'', starting with the profiles $P_2, P_3, P_5, P_6$, and working our way up to conclude that $F(P_0) = \emptyset$, a contradiction. Again let us mirror parts of the proof of Theorem~\ref{thm:odd} to illustrate the idea. Since $d$ is the Condorcet winner at $P_6$, we have $c\not\in F(P_6)$. After reversing three $cdba$ voters, we obtain $P_4$, and must have $c\not\in F(P_6)$, since the optimum according to the $cdba$ voters needs to weakly get worse, by optimistic half-way monotonicity. Similarly since $b$ is the Condorcet winner at $P_5$, we have $d\not\in F(P_5)$. After reversing two $dbac$ voters, we obtain $P_4$ again, and now we see that $d\not\in F(P_4)$, by optimistic half-way monotonicity. So $c,d\not\in F(P_4)$. After reversing two $dcba$ voters in $P_4$, we obtain $P_0$. By optimistic half-way monotonicity, we must have $c,d\not\in F(P_0)$. By similarly following the proof of Theorem~\ref{thm:odd} for $P_1,P_2,P_3$, we can establish that $F(P_0) \cap (\{a,b\}\cup X) = \emptyset$, which gives a contradiction.
\end{proof}

Of course, the resolute Condorcet extensions of Proposition~\ref{prop:positive} induce set-valued Condorcet extensions satisfying both optimistic and pessimistic half-way monotonicity. Hence the bounds of Theorem~\ref{thm:irresolute} are also tight. 

Sanver and Zwicker~\citep[Section 4.2]{SaZw10a} consider a different set extension (the G\"{a}rdenfors extension, also known as Fishburn's extension). They show that a strong version of half-way monotonicity with this set extension is incompatible with Condorcet-consistency (and our impossibilities for $n=15$ and $n=24$ apply to this setting as well), but certain irresolute rules like the top cycle satisfy a weak version. On the other hand, Brandt and Geist~\citep[Footnote~8]{BrGe15a} show that there is no tournament solution (a voting rule depending only on the majority relation) that satisfies this weak version and is also a significant refinement of the top cycle (in the sense of refining the \emph{uncovered set}). This result depends on an interesting observation that, for tournament solutions and for all set extensions, weak half-way monotonicity is equivalent to weak strategyproofness \citep[Brandt and Geist,][Theorem~1]{BrGe15a}.

\subsection{Strong Preference Reversal Paradoxes}

A (resolute) voting rule suffers from the \emph{strong} preference reversal paradox if a voter, by reversing their preferences, can cause their \emph{most}-preferred alternative to win. This is a rather astounding phenomenon -- for example, it could be that if $i$ truthfully votes $a\succ b \succ c \succ d$ the outcome would be $b$, while if $i$ submits $d \succ c \succ b \succ a$ then the outcome would be $a$, which $i$ now ranks last! Thankfully, it is in principle possible for Condorcet extensions to avoid this paradox: maximin is an example. But, in fact, maximin is the \emph{only} of the commonly considered Condorcet extensions which avoids this behaviour, as we will now see. The following argument is an adaptation of P\'erez~\cite{Pere01a}. For the results in this section, we have made no effort to minimise the number of voters required.

Consider five alternatives, $A = \{x,y,z,u,t\}$, and the following 41-voter profile $P_\dagger$:

\begin{center}
	\begin{tabular}{cccccccccc}
		\toprule
		5&7&3&6&1&2&3&5&8&1\\
		\midrule
		$x$&$x$&$y$&$y$&$y$&$y$&$z$&$z$&$u$&$t$ \\
		$z$&$t$&$x$&$x$&$u$&$t$&$y$&$y$&$z$&$y$ \\
		$y$&$u$&$u$&$t$&$z$&$u$&$u$&$t$&$t$&$z$ \\
		$t$&$z$&$z$&$u$&$x$&$x$&$x$&$x$&$y$&$u$ \\
		$u$&$y$&$t$&$z$&$t$&$z$&$t$&$u$&$x$&$x$ \\
		\bottomrule 
	\end{tabular}
\end{center}

\begin{theorem}
	\label{thm:perez}
	If $f$ is a Condorcet extension that avoids the strong preference reversal paradox, then $f(P_\dagger) = t$.
\end{theorem}

\noindent
The punchline is that most popular Condorcet extensions do \emph{not} choose $t$ when faced with~$P_\dagger$: Black chooses $y$ (the Borda winner), the unique Kemeny ranking is $zyxtu$ with $z$ on top, Baldwin and Nanson choose $z$, Dodgson chooses $y$ (in 9 swaps, $t$ takes 15), and the uncovered set is $\{x,y,z\}$, so tie-broken versions of all the common tournament solutions (those that are refinements of the uncovered set, like Copeland, Slater, TEQ, or the bipartisan set) will not select $t$ either. The ``correct choice'' of $t$ is made by maximin, as well as Young's rule, Schulze's method and Ranked Pairs. For the latter three, one can construct other examples where they, too, suffer from the strong preference reversal paradox.

\begin{proof}[Proof of Theorem~\ref{thm:perez}]
	Suppose, for a contradiction, that $f$ is a Condorcet extension that avoids the strong preference reversal paradox, but $f(P_\dagger) \neq t$. 
	\begin{itemize}
		\item If $f(P_\dagger) = x$, then $x$ is also selected after 8 $uztyx$ voters reverse their preferences (to avoid paradox), but in the resulting profile $y$ is Condorcet winner, a contradiction. 
		\item If $f(P_\dagger) = y$, then $y$ is also selected after 7 $xtuzy$ voters reverse their preferences, but in the resulting profile $z$ is Condorcet winner, a contradiction. 
		\item If $f(P_\dagger) = z$, then $z$ is also selected after 6 $yxtuz$ voters reverse their preferences, but in the resulting profile $u$ is Condorcet winner, a contradiction. 
		\item If $f(P_\dagger) = u$, then $u$ is also selected after 5 $xzytu$ voters reverse their preferences, but in the resulting profile $t$ is Condorcet winner, a contradiction.
	\end{itemize}
	Since each case leads to a contradiction, we conclude that $f(P_\dagger) = t$.
\end{proof}

Another approach to the strong preference reversal paradox can be taken by following Duddy's~\cite{Dudd14b} interpretation of the strong no-show paradox. He considers \emph{weak} preferences (allowing indifferences). In this setting, we say that a voting rule suffers from the strong preference reversal paradox if a voter, by reversing their preferences, can cause \emph{one of} their most-preferred alternatives to win. Duddy considered the analogous strong no-show paradox, and showed that all Condorcet extensions suffer from it. One can prove an analogue of this result for the strong preference reversal paradox by `doubling' all voters in Duddy's proof, following the remark after the proof of Theorem~\ref{thm:even}.

\section{A Version of the Gibbard--Satterthwaite Theorem}
\label{sec:zwicker}

The famous Gibbard--Satterthwaite Theorem states that every non-dictatorial voting rule that has full range must be manipulable, when there are at least 3 alternatives. The theorem is somewhat opaque in that it does not tell us \emph{what kinds} of manipulations will be successful. Here we will combine two results to show that every non-trivial voting rule is either \emph{needlessly} or \emph{egregiously} manipulable.

First, let us define some axioms. Fix some electorate $N$, and let $\mathcal D \subseteq A!^N$ be a \emph{domain}, i.e., a subcollection of profiles. A \emph{voting rule on $\mathcal D$} is a map $f : \mathcal D \to A$. We say that
\begin{itemize}
	\item $f$ is \emph{non-imposed} (or is \emph{onto}, or has \emph{full range}) if for all $a\in A$, there is $P\in\mathcal D$ with $f(P) = a$;
	\item $f$ is \emph{non-dictatorial} if there is no $i\in N$ such that $f(P) = \max_{\succ_i} A$ for all profiles $P \in \mathcal{D}$;
	\item $f$ is \emph{unanimous} if for all $a\in A$, we have that $f(P) = a$ for all $P \in \mathcal D$ such that every voter $i\in N$ ranks $a$ in top position;
	\item $f$ is \emph{anonymous} if $f(\sigma P) = f(P)$ for all permutations $\sigma$ of $N$;
	\item $f$ is \emph{manipulable} if there exists a profile $P$, a voter $i$, and a linear order $\succ_i'$ such that both $(P_{-i}, \succ_i') \in \mathcal D$ and $(P_{-i}, \succ_i) \in \mathcal D$, and $f(P_{-i}, \succ_i') \succ_i f(P_{-i}, \succ_i)$. Thus, voter $i$ strictly prefers misrepresenting their preferences.
\end{itemize}
Note that unanimity is stronger than non-imposition, and that anonymity is stronger than non-dictatorship.

Besides the incompatibility between Condorcet-consistency and half-way monotonicity that has been the topic of this paper, the other tool we need is the Campbell--Kelly Theorem. We let $\mathcal{D}_{\text{Condorcet}}$ denote the domain of all profiles $P\in A!^N$ that admit a Condorcet winner.

\begin{theorem}[Campbell and Kelly, \citealp{CaKe03a,CaKe16a}]
	\label{thm:CK}
	Suppose $N$ contains an odd number of voters and $|A| \ge 3$. Let $f : \mathcal{D}_{\text{Condorcet}} \to A$ be an onto and non-dictatorial voting rule. Then $f$ is not manipulable if and only if $f$ is identical to the \emph{Condorcet rule }, i.e., $f(P)$ is the Condorcet winner of $P$ for every $P\in \mathcal{D}_{\text{Condorcet}}$.
\end{theorem}
We will be interested in an implication of the Campbell--Kelly Theorem for voting rules $f : A!^N \to A$ defined for \emph{all} profiles.
\begin{corollary}
	\label{cor:CK}
	Suppose $N$ contains an odd number of voters and $|A| \ge 3$. Let $f : A!^N \to A$ be a voting rule on the full domain. Suppose that $f$ is anonymous and unanimous. If $f$ is \emph{not} Condorcet-consistent, then $f$ is manipulable on $\mathcal{D}_{\text{Condorcet}}$.
\end{corollary}
\begin{proof}
	If $f$ is anonymous, then $f|_{\mathcal{D}_{\text{Condorcet}}}$ is also anonymous and thus non-dictatorial. Similarly, if $f$ is unanimous, then $f|_{\mathcal{D}_{\text{Condorcet}}}$ is non-imposed, since all profiles in which an alternative is ranked top by everyone is contained in $\mathcal{D}_{\text{Condorcet}}$. Thus, by Theorem~\ref{thm:CK}, $f|_{\mathcal{D}_{\text{Condorcet}}}$ is manipulable.
\end{proof}

This implies that if $f$ is \emph{not} a Condorcet extension, then $f$ admits a successful manipulation that occurs between two profiles that have Condorcet winners. Thus, $f$ is, in a sense, \emph{needlessly} manipulable since it could avoid this manipulation if only it selected the Condorcet winners of those profiles. 

On the other hand, we have seen in this paper that Condorcet extensions are \emph{also} guaranteed to be manipulable (when $n$ and $m$ are large enough) by Theorem~\ref{thm:odd}. Precisely, they all suffer from the preference reversal paradox, instances of which we might describe as \emph{egregious} manipulation instances. Putting these two results together, we obtain a ``disjunctive Gibbbard--Satterthwaite theorem'' that uses slightly stronger assumptions,\footnote{Namely, higher bounds on $n$ and $m$, the requirement that $n$ is odd, anonymity instead of non-dictatorship, and unanimity instead of non-imposition. The latter two assumptions can evidently be weakened to non-dictatorship and non-imposition on $\mathcal{D}_{\text{Condorcet}}$, but this makes the statement seem less appealing.} but more explicitly identifies the nature of manipulability. This disjunctive version was first proposed by Zwicker~\citep[Corollary~2.8]{Zwic15a} using slightly different assumptions.

\begin{corollary}[Zwicker's Corollary]
	\label{cor:zwicker}
	Suppose there are at least 4 alternatives and an odd number\footnote{Because Theorem~\ref{thm:CK} only holds for odd electorates, this corollary also requires this assumption. Campbell and Kelly~\cite{CaKe15a} show that their theorem also holds for even electorates if we require the stronger axioms of anonymity and neutrality (on $\mathcal{D}_{\text{Condorcet}}$). However, on the full domain, anonymity and neutrality are usually incompatible with resoluteness. It would be interesting to find appealing conditions that allow an analogue of Corollary~\ref{cor:zwicker} to go through for even electorates as well.} of at least 15 voters. Let $f$ be an anonymous and unanimous voting rule. Either $f$ is manipulable on $\mathcal{D}_{\text{Condorcet}}$, or $f$ is manipulable by preference reversal.
\end{corollary}

\begin{proof}
	If $f$ is a Condorcet extension then this follows from Theorem~\ref{thm:odd}; otherwise it follows from Corollary~\ref{cor:CK}.
\end{proof}

\section{Conclusions and Future Work}

In this paper, we have undertaken a detailed study of the preference reversal paradox. We have seen that many known results about the no-show paradox transfer to our setting, but that these impossibilities require additional voters to go through. An interesting contrast appeared for set-valued rules: while imposing optimistic or pessimistic \emph{participation} allows for stronger positive results (in terms of number of voters supported), we saw that optimistic or pessimistic half-way monotonicity is essentially as strong as its resolute version, not allowing for additional voters. Our computer-aided approach leaves some questions unresolved. In particular, our positive results are somewhat frustrating, since the voting rules produced are only available as lookup tables, and it would be desirable to find concise descriptions. Moreover, the approach does not extend well for $m\ge 5$ alternatives, since searching over profiles of larger size quickly becomes infeasible; thus, we do not know if our impossibilities are proveable with fewer voters as $m$ increases. As we have seen, half-way monotonicity is weaker than participation; it would be interesting to find natural voting rules that satisfy half-way monotonicity but fail participation. It appears that the only known natural rules satisfying half-way monotonicity are scoring rules; some other artificial rules are discussed by Moulin~\citep[p.~63]{Moul88b}, Campbell and Kelly~\cite{CaKe02c}, and N\'u\~nez and Sanver~\cite{NuSa17a}.

Our disjunctive Gibbard--Satterthwaite theorem gives a more explicit account of the types of manipulations that are needed to prove it. A related result appears in the literature: the Gibbard--Satterthwaite Theorem holds even if we allow voters to report only preferences that are obtainable by at most one (adjacent) swap from their honest vote [Caragiannis et al., \citealp{CESY12a}, Sato, \citealp{Sato13b}]. It is plausible that there are further interesting results of this sort; gaining a better understanding of them could complement the literature on \emph{dictatorial domains} \citep[Aswal et al.,][]{ACS03a} which studies restricted preference domains that still lead to a Gibbard--Satterthwaite-style impossibility.

Given our results about set-valued voting rules in Section~\ref{sec:irresolute}, one may hope to get an analogue of the Duggan--Schwartz Theorem [\citealp{DuSc00a}, see also Taylor, \citealp{Tayl05a}] in the style of Corollary~\ref{cor:zwicker}. The Duggan-Schwartz Theorem states that any (non-trivial) set-valued voting rule that is not manipulable by optimists or pessimists must have a \emph{nominator}, that is, a voter whose top choice is always part of the returned choice set. For set-valued Condorcet extensions (which cannot have nominators), we know from Theorem~\ref{thm:irresolute} that they cannot satisfy optimistic or pessimistic half-way monotonicity. However, the Campbell--Kelly Theorem does not admit a direct analogue to the set-valued context, because rules with nominators are also strategyproof on $\mathcal{D}_{\text{Condorcet}}$. It would be interesting to obtain a justification for Condorcet-consistency in the style of the Campbell--Kelly Theorem for set-valued voting rules, perhaps by adding additional axioms.

\paragraph{Acknowledgements}
	I thank Felix Brandt, Christian Geist, and Bill Zwicker for discussions that led to this paper. I am indebted to the anonymous reviewers for excellent suggestions. I am supported by EPSRC, by ERC under grant number 639945 (ACCORD), and by COST Action IC1205.

\bibliographystyle{eptcs}

\begin{thebibliography}{10}
	\providecommand{\bibitemdeclare}[2]{}
	\providecommand{\surnamestart}{}
	\providecommand{\surnameend}{}
	\providecommand{\urlprefix}{Available at }
	\providecommand{\url}[1]{\texttt{#1}}
	\providecommand{\href}[2]{\texttt{#2}}
	\providecommand{\urlalt}[2]{\href{#1}{#2}}
	\providecommand{\doi}[1]{doi:\urlalt{http://dx.doi.org/#1}{#1}}
	\providecommand{\bibinfo}[2]{#2}
	
	\bibitemdeclare{article}{ACS03a}
	\bibitem{ACS03a}
	\bibinfo{author}{N.~\surnamestart Aswal\surnameend},
	\bibinfo{author}{S.~\surnamestart Chatterji\surnameend} \&
	\bibinfo{author}{A.~\surnamestart Sen\surnameend} (\bibinfo{year}{2003}):
	\emph{\bibinfo{title}{Dictatorial domains}}.
	\newblock {\sl \bibinfo{journal}{Economic Theory}}
	\bibinfo{volume}{22}(\bibinfo{number}{1}), pp. \bibinfo{pages}{45--62},
	\doi{10.1007/s00199-002-0285-8}.
	
	\bibitemdeclare{inproceedings}{AuSi09a}
	\bibitem{AuSi09a}
	\bibinfo{author}{Gilles \surnamestart Audemard\surnameend} \&
	\bibinfo{author}{Laurent \surnamestart Simon\surnameend}
	(\bibinfo{year}{2009}): \emph{\bibinfo{title}{Predicting Learnt Clauses
			Quality in Modern {SAT} Solvers}}.
	\newblock In: {\sl \bibinfo{booktitle}{Proceedings of the 21st International
			Joint Conference on Artificial Intelligence (IJCAI)}}, pp.
	\bibinfo{pages}{399--404}.
	
	\bibitemdeclare{article}{BeMa12a}
	\bibitem{BeMa12a}
	\bibinfo{author}{A.~\surnamestart Belov\surnameend} \&
	\bibinfo{author}{J.~\surnamestart Marques-Silva\surnameend}
	(\bibinfo{year}{2012}): \emph{\bibinfo{title}{{MUSer2}: An efficient {MUS}
			extractor}}.
	\newblock {\sl \bibinfo{journal}{Journal on Satisfiability, Boolean Modeling
			and Computation}} \bibinfo{volume}{8}, pp. \bibinfo{pages}{123--128}.
	
	\bibitemdeclare{inproceedings}{Bier13a}
	\bibitem{Bier13a}
	\bibinfo{author}{A.~\surnamestart Biere\surnameend} (\bibinfo{year}{2013}):
	\emph{\bibinfo{title}{{L}ingeling, {P}lingeling and {T}reengeling entering
			the {SAT} competition 2013}}.
	\newblock In: {\sl \bibinfo{booktitle}{Proceedings of the SAT Competition
			2013}}, pp. \bibinfo{pages}{51--52}.
	
	\bibitemdeclare{inproceedings}{BBG15a}
	\bibitem{BBG15a}
	\bibinfo{author}{F.~\surnamestart Brandl\surnameend},
	\bibinfo{author}{F.~\surnamestart Brandt\surnameend} \&
	\bibinfo{author}{C.~\surnamestart Geist\surnameend} (\bibinfo{year}{2016}):
	\emph{\bibinfo{title}{Proving the Incompatibility of Efficiency and
			Strategyproofness via {SMT} Solving}}.
	\newblock In: {\sl \bibinfo{booktitle}{Proceedings of the 25th International
			Joint Conference on Artificial Intelligence (IJCAI)}},
	\bibinfo{publisher}{AAAI Press}, pp. \bibinfo{pages}{116--122}.
	
	\bibitemdeclare{inproceedings}{BBGH15a}
	\bibitem{BBGH15a}
	\bibinfo{author}{F.~\surnamestart Brandl\surnameend},
	\bibinfo{author}{F.~\surnamestart Brandt\surnameend},
	\bibinfo{author}{C.~\surnamestart Geist\surnameend} \&
	\bibinfo{author}{J.~\surnamestart Hofbauer\surnameend}
	(\bibinfo{year}{2015}): \emph{\bibinfo{title}{Strategic Abstention based on
			Preference Extensions: {P}ositive Results and Computer-Generated
			Impossibilities}}.
	\newblock In: {\sl \bibinfo{booktitle}{Proceedings of the 24th International
			Joint Conference on Artificial Intelligence (IJCAI)}},
	\bibinfo{publisher}{AAAI Press}, pp. \bibinfo{pages}{18--24}.
	
	\bibitemdeclare{article}{BrGe15a}
	\bibitem{BrGe15a}
	\bibinfo{author}{F.~\surnamestart Brandt\surnameend} \&
	\bibinfo{author}{C.~\surnamestart Geist\surnameend} (\bibinfo{year}{2016}):
	\emph{\bibinfo{title}{Finding Strategyproof Social Choice Functions via {SAT}
			Solving}}.
	\newblock {\sl \bibinfo{journal}{Journal of Artificial Intelligence Research}}
	\bibinfo{volume}{55}, pp. \bibinfo{pages}{565--602}, \doi{10.1613/jair.4959}.
	
	\bibitemdeclare{inproceedings}{BGP16c}
	\bibitem{BGP16c}
	\bibinfo{author}{F.~\surnamestart Brandt\surnameend},
	\bibinfo{author}{C.~\surnamestart Geist\surnameend} \&
	\bibinfo{author}{D.~\surnamestart Peters\surnameend} (\bibinfo{year}{2016}):
	\emph{\bibinfo{title}{Optimal Bounds for the No-Show Paradox via {SAT}
			Solving}}.
	\newblock In: {\sl \bibinfo{booktitle}{Proceedings of the 15th International
			Conference on Autonomous Agents and Multiagent Systems (AAMAS)}}, pp.
	\bibinfo{pages}{314--322}.
	
	\bibitemdeclare{article}{CaKe02c}
	\bibitem{CaKe02c}
	\bibinfo{author}{D.~E. \surnamestart Campbell\surnameend} \&
	\bibinfo{author}{J.~S. \surnamestart Kelly\surnameend}
	(\bibinfo{year}{2002}): \emph{\bibinfo{title}{Non-monotonicity does not imply
			the no-show paradox}}.
	\newblock {\sl \bibinfo{journal}{Social Choice and Welfare}}
	\bibinfo{volume}{19}(\bibinfo{number}{3}), pp. \bibinfo{pages}{513--515},
	\doi{10.1007/s003550100128}.
	
	\bibitemdeclare{article}{CaKe03a}
	\bibitem{CaKe03a}
	\bibinfo{author}{D.~E. \surnamestart Campbell\surnameend} \&
	\bibinfo{author}{J.~S. \surnamestart Kelly\surnameend}
	(\bibinfo{year}{2003}): \emph{\bibinfo{title}{A strategy-proofness
			characterization of majority rule}}.
	\newblock {\sl \bibinfo{journal}{Economic Theory}}
	\bibinfo{volume}{22}(\bibinfo{number}{3}), pp. \bibinfo{pages}{557--568},
	\doi{10.1007/s00199-002-0344-1}.
	
	\bibitemdeclare{article}{CaKe15a}
	\bibitem{CaKe15a}
	\bibinfo{author}{D.~E. \surnamestart Campbell\surnameend} \&
	\bibinfo{author}{J.~S. \surnamestart Kelly\surnameend}
	(\bibinfo{year}{2015}): \emph{\bibinfo{title}{Anonymous, neutral, and
			strategy-proof rules on the {Condorcet} domain}}.
	\newblock {\sl \bibinfo{journal}{Economics Letters}} \bibinfo{volume}{128}, pp.
	\bibinfo{pages}{79--82}, \doi{10.1016/j.econlet.2015.01.009}.
	
	\bibitemdeclare{article}{CaKe16a}
	\bibitem{CaKe16a}
	\bibinfo{author}{D.~E. \surnamestart Campbell\surnameend} \&
	\bibinfo{author}{J.~S. \surnamestart Kelly\surnameend}
	(\bibinfo{year}{2016}): \emph{\bibinfo{title}{Correction to ``{A}
			Strategy-proofness Characterization of Majority Rule''}}.
	\newblock {\sl \bibinfo{journal}{Economic Theory Bulletin}}
	\bibinfo{volume}{4}(\bibinfo{number}{1}), pp. \bibinfo{pages}{121--124},
	\doi{10.1007/s40505-015-0066-8}.
	
	\bibitemdeclare{inproceedings}{CESY12a}
	\bibitem{CESY12a}
	\bibinfo{author}{I.~\surnamestart Caragiannis\surnameend},
	\bibinfo{author}{E.~\surnamestart Elkind\surnameend},
	\bibinfo{author}{M.~\surnamestart Szegedy\surnameend} \&
	\bibinfo{author}{L.~\surnamestart Yu\surnameend} (\bibinfo{year}{2012}):
	\emph{\bibinfo{title}{Mechanism design: from partial to probabilistic
			verification}}.
	\newblock In: {\sl \bibinfo{booktitle}{Proceedings of the 13th ACM Conference
			on Electronic Commerce (ACM EC)}}, \bibinfo{organization}{ACM}, pp.
	\bibinfo{pages}{266--283}, \doi{10.1145/2229012.2229035}.
	
	\bibitemdeclare{article}{Dudd14b}
	\bibitem{Dudd14b}
	\bibinfo{author}{C.~\surnamestart Duddy\surnameend} (\bibinfo{year}{2014}):
	\emph{\bibinfo{title}{Condorcet's principle and the strong no-show
			paradoxes}}.
	\newblock {\sl \bibinfo{journal}{Theory and Decision}}
	\bibinfo{volume}{77}(\bibinfo{number}{2}), pp. \bibinfo{pages}{275--285},
	\doi{10.1007/s11238-013-9401-4}.
	
	\bibitemdeclare{article}{DuSc00a}
	\bibitem{DuSc00a}
	\bibinfo{author}{J.~\surnamestart Duggan\surnameend} \&
	\bibinfo{author}{T.~\surnamestart Schwartz\surnameend}
	(\bibinfo{year}{2000}): \emph{\bibinfo{title}{Strategic Manipulability
			without Resoluteness or Shared Beliefs: {G}ibbard-{S}atterthwaite
			Generalized}}.
	\newblock {\sl \bibinfo{journal}{Social Choice and Welfare}}
	\bibinfo{volume}{17}(\bibinfo{number}{1}), pp. \bibinfo{pages}{85--93},
	\doi{10.1007/PL00007177}.
	
	\bibitemdeclare{article}{BrFi83a}
	\bibitem{BrFi83a}
	\bibinfo{author}{P.~C. \surnamestart Fishburn\surnameend} \&
	\bibinfo{author}{S.~J. \surnamestart Brams\surnameend}
	(\bibinfo{year}{1983}): \emph{\bibinfo{title}{Paradoxes of Preferential
			Voting}}.
	\newblock {\sl \bibinfo{journal}{Mathematics Magazine}}
	\bibinfo{volume}{56}(\bibinfo{number}{4}), pp. \bibinfo{pages}{207--214},
	\doi{10.2307/2689808}.
	
	\bibitemdeclare{article}{GeEn11a}
	\bibitem{GeEn11a}
	\bibinfo{author}{C.~\surnamestart Geist\surnameend} \&
	\bibinfo{author}{U.~\surnamestart Endriss\surnameend} (\bibinfo{year}{2011}):
	\emph{\bibinfo{title}{Automated Search for Impossibility Theorems in Social
			Choice Theory: Ranking Sets of Objects}}.
	\newblock {\sl \bibinfo{journal}{Journal of Artificial Intelligence Research}}
	\bibinfo{volume}{40}, pp. \bibinfo{pages}{143--174}, \doi{10.1613/jair.3126}.
	
	\bibitemdeclare{incollection}{GePe17a}
	\bibitem{GePe17a}
	\bibinfo{author}{C.~\surnamestart Geist\surnameend} \&
	\bibinfo{author}{D.~\surnamestart Peters\surnameend} (\bibinfo{year}{2017}):
	\emph{\bibinfo{title}{Computer-aided Methods for Social Choice Theory}}.
	\newblock In \bibinfo{editor}{U.~\surnamestart Endriss\surnameend}, editor:
	{\sl \bibinfo{booktitle}{Trends in Computational Social Choice}},
	chapter~\bibinfo{chapter}{13}.
	\newblock \bibinfo{note}{Forthcoming}.
	
	\bibitemdeclare{inproceedings}{HKM16a}
	\bibitem{HKM16a}
	\bibinfo{author}{M.~J.~H. \surnamestart Heule\surnameend},
	\bibinfo{author}{O.~\surnamestart Kullmann\surnameend} \&
	\bibinfo{author}{V.~W. \surnamestart Marek\surnameend}
	(\bibinfo{year}{2016}): \emph{\bibinfo{title}{Solving and Verifying the
			{B}oolean {P}ythagorean Triples Problem via Cube-and-Conquer}}.
	\newblock In: {\sl \bibinfo{booktitle}{Proceedings of the 19th International
			Conference on Theory and Applications of Satisfiability Testing}}, {\sl
		\bibinfo{series}{Lecture Notes in Computer Science (LNCS)}}
	\bibinfo{volume}{9710}, \bibinfo{publisher}{Springer-Verlag}, pp.
	\bibinfo{pages}{228--245}, \doi{10.1007/978-3-662-48899-7_31}.
	
	\bibitemdeclare{article}{JPG09a}
	\bibitem{JPG09a}
	\bibinfo{author}{J.~L. \surnamestart Jimeno\surnameend},
	\bibinfo{author}{J.~\surnamestart P{\'e}rez\surnameend} \&
	\bibinfo{author}{E.~\surnamestart Garc{\'\i}a\surnameend}
	(\bibinfo{year}{2009}): \emph{\bibinfo{title}{An extension of the {M}oulin
			{N}o {S}how {P}aradox for voting correspondences}}.
	\newblock {\sl \bibinfo{journal}{Social Choice and Welfare}}
	\bibinfo{volume}{33}(\bibinfo{number}{3}), pp. \bibinfo{pages}{343--459},
	\doi{10.1007/s00355-008-0360-6}.
	
	\bibitemdeclare{article}{LPMM15a}
	\bibitem{LPMM15a}
	\bibinfo{author}{M.~H. \surnamestart Liffiton\surnameend},
	\bibinfo{author}{A.~\surnamestart Previti\surnameend},
	\bibinfo{author}{A.~\surnamestart Malik\surnameend} \&
	\bibinfo{author}{J.~\surnamestart Marques-Silva\surnameend}
	(\bibinfo{year}{2015}): \emph{\bibinfo{title}{Fast, flexible {MUS}
			enumeration}}.
	\newblock {\sl \bibinfo{journal}{Constraints}}, pp. \bibinfo{pages}{1--28},
	\doi{10.1007/s10601-015-9183-0}.
	
	\bibitemdeclare{article}{Moul88b}
	\bibitem{Moul88b}
	\bibinfo{author}{H.~\surnamestart Moulin\surnameend} (\bibinfo{year}{1988}):
	\emph{\bibinfo{title}{Condorcet's Principle implies the No Show Paradox}}.
	\newblock {\sl \bibinfo{journal}{Journal of Economic Theory}}
	\bibinfo{volume}{45}, pp. \bibinfo{pages}{53--64},
	\doi{10.1016/0022-0531(88)90253-0}.
	
	\bibitemdeclare{article}{NuSa17a}
	\bibitem{NuSa17a}
	\bibinfo{author}{M.~\surnamestart N{\'u}{\~n}ez\surnameend} \&
	\bibinfo{author}{M.~R. \surnamestart Sanver\surnameend}
	(\bibinfo{year}{2017}): \emph{\bibinfo{title}{Revisiting the connection
			between the no-show paradox and monotonicity}}.
	\newblock {\sl \bibinfo{journal}{Mathematical Social Sciences}},
	\doi{10.1016/j.mathsocsci.2017.02.003}.
	\newblock \bibinfo{note}{Forthcoming}.
	
	\bibitemdeclare{article}{Pere01a}
	\bibitem{Pere01a}
	\bibinfo{author}{J.~\surnamestart P{\'e}rez\surnameend} (\bibinfo{year}{2001}):
	\emph{\bibinfo{title}{The Strong No Show Paradoxes are a common flaw in
			{C}ondorcet voting correspondences}}.
	\newblock {\sl \bibinfo{journal}{Social Choice and Welfare}}
	\bibinfo{volume}{18}(\bibinfo{number}{3}), pp. \bibinfo{pages}{601--616},
	\doi{10.1007/s003550000079}.
	
	\bibitemdeclare{article}{SaZw09a}
	\bibitem{SaZw09a}
	\bibinfo{author}{M.~R. \surnamestart Sanver\surnameend} \&
	\bibinfo{author}{W.~S. \surnamestart Zwicker\surnameend}
	(\bibinfo{year}{2009}): \emph{\bibinfo{title}{One-way monotonicity as a form
			of strategy-proofness}}.
	\newblock {\sl \bibinfo{journal}{International Journal of Game Theory}}
	\bibinfo{volume}{38}(\bibinfo{number}{4}), pp. \bibinfo{pages}{553--574},
	\doi{10.1007/s00182-009-0170-9}.
	
	\bibitemdeclare{article}{SaZw10a}
	\bibitem{SaZw10a}
	\bibinfo{author}{M.~R. \surnamestart Sanver\surnameend} \&
	\bibinfo{author}{W.~S. \surnamestart Zwicker\surnameend}
	(\bibinfo{year}{2012}): \emph{\bibinfo{title}{Monotonicity properties and
			their adaption to irresolute social choice rules}}.
	\newblock {\sl \bibinfo{journal}{Social Choice and Welfare}}
	\bibinfo{volume}{39}(\bibinfo{number}{2--3}), pp. \bibinfo{pages}{371--398},
	\doi{10.1007/s00355-012-0654-6}.
	
	\bibitemdeclare{article}{Sato13b}
	\bibitem{Sato13b}
	\bibinfo{author}{S.~\surnamestart Sato\surnameend} (\bibinfo{year}{2013}):
	\emph{\bibinfo{title}{A sufficient condition for the equivalence of
			strategy-proofness and non-manipulability by preferences adjacent to the
			sincere one}}.
	\newblock {\sl \bibinfo{journal}{Journal of Economic Theory}}
	\bibinfo{volume}{148}, pp. \bibinfo{pages}{259--278},
	\doi{10.1016/j.jet.2012.12.001}.
	
	\bibitemdeclare{article}{TaLi09a}
	\bibitem{TaLi09a}
	\bibinfo{author}{P.~\surnamestart Tang\surnameend} \&
	\bibinfo{author}{F.~\surnamestart Lin\surnameend} (\bibinfo{year}{2009}):
	\emph{\bibinfo{title}{Computer-aided proofs of {Arrow's} and other
			impossibility theorems}}.
	\newblock {\sl \bibinfo{journal}{Artificial Intelligence}}
	\bibinfo{volume}{173}(\bibinfo{number}{11}), pp. \bibinfo{pages}{1041--1053},
	\doi{10.1016/j.artint.2009.02.005}.
	
	\bibitemdeclare{book}{Tayl05a}
	\bibitem{Tayl05a}
	\bibinfo{author}{A.~D. \surnamestart Taylor\surnameend} (\bibinfo{year}{2005}):
	\emph{\bibinfo{title}{Social Choice and the Mathematics of Manipulation}}.
	\newblock \bibinfo{publisher}{Cambridge University Press},
	\doi{10.1017/CBO9780511614316}.
	
	\bibitemdeclare{incollection}{Zwic15a}
	\bibitem{Zwic15a}
	\bibinfo{author}{W.~S. \surnamestart Zwicker\surnameend}
	(\bibinfo{year}{2016}): \emph{\bibinfo{title}{Introduction to the Theory of
			Voting}}.
	\newblock In \bibinfo{editor}{F.~\surnamestart Brandt\surnameend},
	\bibinfo{editor}{V.~\surnamestart Conitzer\surnameend},
	\bibinfo{editor}{U.~\surnamestart Endriss\surnameend},
	\bibinfo{editor}{J.~\surnamestart Lang\surnameend} \& \bibinfo{editor}{A.~D.
		\surnamestart Procaccia\surnameend}, editors: {\sl
		\bibinfo{booktitle}{Handbook of Computational Social Choice}},
	chapter~\bibinfo{chapter}{2}, \bibinfo{publisher}{Cambridge University
		Press}, \doi{10.1017/CBO9781107446984.003}.
	
\end{thebibliography}

\end{document}